%% file: diamond_saturation.tex
\documentclass[journal]{IEEEtran}

\usepackage[T1]{fontenc} 
\usepackage[english]{babel}

\usepackage[margin=2.4cm]{geometry}

\usepackage[pdftex]{graphicx}
\ifCLASSOPTIONcompsoc
  \usepackage[caption=false,font=normalsize,labelfont=sf,textfont=sf]{subfig}
\else
  \usepackage[caption=false,font=footnotesize]{subfig}
\fi
\usepackage[colorlinks=true]{hyperref}
\usepackage[table,usenames,dvipsnames]{xcolor}

\usepackage{url}

\usepackage{amssymb} 
\usepackage{mathtools} 
\usepackage{amsmath} 
\usepackage{amsfonts}
\usepackage{dsfont}
\usepackage{mathrsfs} 
\usepackage{paralist}
\usepackage{enumitem}
\usepackage{amsthm}
\usepackage{complexity}

\newtheorem{theorem}{Theorem}
\newtheorem{definition}[theorem]{Definition}
\newtheorem{proposition}[theorem]{Proposition}

\newtheorem{observation}[theorem]{Observation}

\input{mymath.tex}

\newcommand{\uc}{Institute for Theoretical Physics, University of Cologne, Germany}
\newcommand{\ug}{Institute of Theoretical Physics and Astrophysics, University of Gda\'{n}sk, Poland}
\newcommand{\syd}{Centre for Engineered Quantum Systems, School of Physics,
The University of Sydney, Australia}
\newcommand{\calt}{Institute for Quantum information and Matter, California Institute of Technology, Pasadena, USA}

\newcommand{\hhu}{Institute for Theoretical Physics,
	Heinrich Heine University D{\"u}sseldorf, Germany}

\title{Note on the saturation of the norm inequalities between diamond and nuclear norm}

\author{Ulrich Michel,
		Martin Kliesch,
		Richard Kueng, 
		and David Gross
	\thanks{
	All authors are with the \uc{}.
	M.\ Kliesch is also with the \ug{} and moved during the final publication process to the \hhu. 
	R.\ Kueng is now affiliated with the \calt{} and D.\ Gross with the \syd.
	\newline
	E-mails: ulrich.peer.michel@gmail.com, science@mkliesch.eu, rkueng@caltech.edu, david.gross@thp.uni-koeln.de. 
	}
	\thanks{
	The work of U.\ Michel was supported by the BCGS under Grant GSC 260. 
	The work of M.\ Kliesch was funded by the National Science Centre, Poland (Polonez 2015/19/P/ST2/03001) within the European Union's Horizon 2020 research and innovation programme under the Marie Sk{\l}odowska-Curie grant agreement No 665778.
	R.\ Kueng and D.\ Gross acknowledge financial support from the Excellence Initiative of the German 
	Federal and State Governments (Grant ZUK 81), the ARO under contract 
	W911NF-14-1-0098 (Quantum Characterization, Verification, and 
	Validation), and the DFG (SPP1798 CoSIP).
	}
}

\makeatletter 
 \hypersetup{pdftitle = {Saturation of the norm inequalities between diamond and nuclear norm},
	     pdfauthor = {Ulrich Michel, Martin Kliesch, Richard Kueng, and David Gross},
	     pdfsubject = {Quantum information, compressed sensing, compressive sensing},
	     pdfkeywords = {Low-rank matrix recovery, convex relaxation, quantum process tomography, Hölder's inequality}
	    }
\makeatother

\begin{document}
\maketitle


\begin{abstract}
The diamond norm plays an important role in quantum information and operator theory. 
Recently, it has also been proposed as a regularizer for low-rank matrix recovery.
The norm constants that bound the diamond norm in terms of the nuclear norm (also known as trace norm) are explicitly known.
This note provides a simple characterization of all operators saturating the upper and the lower bound. 
\end{abstract}

\thispagestyle{empty}

\section{Introduction}
The \emph{diamond norm} \cite{KitSheVya02,Wat11} plays an important role in quantum information and in operator theory, where it equals -- essentially\footnote{
	The diamond norm of a map is the  cb-norm of its adjoint $\dnorm{X}=\|X^\dagger\|_{\mathrm{cb}}$ \cite{Wat12,Pau02}. 
} --
the \emph{norm of complete boundedness} or \emph{cb-norm}.
Recently, some of the present authors proposed its use as a regularizer for low-rank matrix recovery in the context of compressed sensing \cite{KliKueEis16}.
Given few random linear measurements on a low rank matrix, traditionally one reconstructs it by minimizing the nuclear norm subject to the constraints given by the measurements. 
Kliesch et al.~\cite{KliKueEis16} recently showed that replacing the nuclear norm by the diamond norm can be beneficial for the reconstruction. 
In particular, it was proven that for operators saturating the norm inequality from Theorem~\ref{thm:lowerBound} below, diamond norm regularization performs at least as well as (and heuristically better than) the usual nuclear norm regularization \cite{KliKueEis16}. 
The set of operators saturating this norm inequality was explicitly characterized \cite{KliKueEis16} (see Theorem~\ref{thm:lowerBound}). 
This characterization has turned out to also follow from an analysis of typical distances between quantum channels \cite{NecPucPaw16Article}. 

These previous analyses are based on a semidefinite programming (SDP) formulation \cite{Wat12} for the diamond norm---henceforth called \emph{Watrous' SDP}. 
The SDP-formulation renders the proof quite technical.

In this note we present a simple and compact proof of this characterization. 
Moreover, we also characterize all matrices saturating a converse norm inequality. 

\subsection{Notation}
The space of linear maps on $\V$ is denoted by $\L\myleft(\V \myright)$ and the identity by 
$\1_\V \in \L(\V)$.
The \emph{tensor products} of operators or vector spaces are both denoted by $\otimes$. 
A nonnegative real number $\sigma$ is called a \emph{singular value} of an operator $A\in \L(\V)$ if and only if,
there exist normalized vectors $u,v \in \V$ such that $Av=\sigma u$. 
This implies $A^\dagger u =\sigma v$ for the \emph{adjoint} operator $A^\dagger$. 
We call such vectors $u,v$ \emph{left singular} and \emph{right singular vector} for $\sigma$, respectively. 
The \emph{operator square root} of a positive semidefinite operator $A \succeq 0$ is defined to be the unique positive semidefinite operator $\sqrt{A}$ obeying $\sqrt{A}^{\,2} = A$. 
We will use the nuclear norm $\EinsNorm{\, \cdot \,}$, Frobenius norm $\ZweiNorm{\, \cdot \,}$ and  spectral norm $\InfNorm{\, \cdot \,}$ defined by
\begin{align}
\EinsNorm{A} &\coloneqq \Tr\sqrt{A^\dagger A} \, , 
\\
\ZweiNorm{A} &\coloneqq \sqrt{\Tr \myleft[ A^\dagger A\myright]} \, ,
\\
\InfNorm{A}  &\coloneqq \sup_{x\in \V }\frac{\ZweiNorm{Ax}}{\ZweiNorm{x}} \, .
\end{align}
Following \cite{KliKueEis16}, we define the \emph{square norm} of a \emph{bipartite operator} 
$X \in \L(\W \otimes \V)$ 
to be
\begin{equation}\label{eq:DiaNorm}
\DiaNorm{X} 
=
\sup_{\substack{
			A,B \in \L\left(\V \right)\\
			\ZweiNorm{A}=\ZweiNorm{B}=\sqrt{\dim(\V)} }}
		\EinsNormb{ \left(\1_{\W }\otimes A\right) X \left( \1_{\W } \otimes B \right)}  \, .
\end{equation}
The square norm is closely related \cite{KliKueEis16} to the diamond norm $\dnorm{\argdot}$ by 
$\DiaNorm{J(M)} = \dim(\V)\dnorm{M}$, where $M : \L(\V) \to \L(\W)$ is a linear map, $J$ is the Choi-Jamio{\l}kowski isomorphism \cite{Cho75,Jam72} (so that $J(M) \in \L(\W\otimes \V)$). 
Finally, the partial trace $\Tr_\W : \L(\W \otimes \V) \to \L(\V)$ is the linear map defined on product operators
$(A\otimes B) \in \L(\W \otimes \V)$ by $\Tr_{\W }[A\otimes B]= \Tr\myleft[A\myright] B$ and linearly extended to all $\L(\W \otimes \V)$.

\section{Result}
The following equivalence of norms holds (see, e.g., \cite{KliKueEis16}):
\begin{equation}
	\EinsNorm{X} \leq \DiaNorm{X} \leq \dim(\V) \EinsNorm{X} \, .
\end{equation}
We provide a characterization of (i) all operators saturating the upper bound and (ii) the lower bound on $\DiaNorm{X}$. 

\begin{theorem}[Saturation of the lower bound {\cite{KliKueEis16}}]
\hfill\\
Let $X \in \L\left(\W \otimes \V \right)$. 
Then 
$ \EinsNorm{X} = \DiaNorm{X}$ 
if and only if
\begin{equation}
\Tr_\W \myleft[\sqrt{XX^\dagger}\myright]=\Tr_\W \myleft[\sqrt{X^\dagger X}\myright] =\frac{\EinsNorm{X}}{\dim(\V)} \1_\V \, .
\end{equation}
\label{thm:lowerBound}
\end{theorem}

This theorem was proven \cite{KliKueEis16} using Watrous' SDP \cite{Wat12} and complementary slackness. 
It also follows from the proof of another new norm inequality \cite[Remark~3]{NecPucPaw16Article}, also based on Watrous' SDP. 

Below we provide a shorter and more direct proof. 
The similar upper bound is saturated as follows. 

\begin{theorem}[Saturation of the upper bound]\hfill\\
Let $X \in \L\left(\W \otimes \V \right)$. Then $  \DiaNorm{X} = \dim(\V)\EinsNorm{X} $ if and only if there exist rank-$1$ projectors 
$\phi\phi^\dagger, \psi\psi^\dagger \in \L\left(\V \right)$ 
such that 
\begin{align}
\Tr_\W \myleft[\sqrt{XX^\dagger}\myright]
&=
\EinsNorm{X} \psi\psi^\dagger 
\quad \text{and} \quad \\
\Tr_\W \myleft[\sqrt{X^\dagger X}\myright] 
&=
\EinsNorm{X} \phi\phi^\dagger \, .
\end{align}
\label{thm:upperBound}
\end{theorem}

\section{Proofs}
In this section we first state our auxiliary statements and then the proofs of our theorems. 

\subsection{Auxiliary statements}
The concept of \emph{sign functions} of real numbers is extendable to non-Hermitian matrices.

\begin{definition}[Sign matrix] For any matrix $X \in \L(\V)$ with singular value decomposition $X = U\Sigma V^\dagger$ we define its \emph{sign matrix} to be $S_X\coloneqq VU^\dagger$.
\end{definition}
Note that the sign matrix obeys
\begin{align}
X S_X &= U\Sigma U^\dagger=\sqrt{XX^\dagger} \, ,
\\
S_X X &=  V\Sigma V^\dagger = \sqrt{X^\dagger X} \, .
\end{align}
The eigenvectors of $X S_X$ and $S_X X$ are \emph{left singular} and \emph{right singular vectors} of $X$, respectively. 
The saturation results from our theorems are derived using similar saturation results for Schatten norms. 
For any matrix $X \in \L(\V)$ with $n \coloneqq \dim(\V)$ it holds that $\EinsNorm{X} \leq \sqrt n \ZweiNorm{X}$ and we observe that this bound is saturated for unitary operators. 

\begin{observation}[Saturation: $1$ and $2$-norm]\label{lem:Norm}
\hfill\\
Let $X \in \L(\V)$. 
Then $\ZweiNorm{X} = \frac{1}{\sqrt{n}}\EinsNorm{X}$ if and only if $\frac{\sqrt{n}}{\EinsNorm{X}} X$ is unitary. 
\end{observation}

Schatten norms satisfy H\"older's inequality, i.e., for $A,B \in \L\left(\V\right)$ it holds that $\EinsNorm{AB} \le \InfNorm{A}\EinsNorm{B}$. 
The saturation condition involves the range $\ran(B)$ of $B$.

\begin{proposition}[Saturation of H\"older's inequality]\label{Obs:HoelderSat1}
Let $A,B \in \L(\V)$. 
Then the following are equivalent:
  \begin{compactenum}[i)]
	  \item\label{item:tightness}
	  Tightness in H\"older's inequality, i.e.\ $\EinsNorm{AB} = \InfNorm{A}\EinsNorm{B}$  
	  \item\label{item:every}
	  Every vector $x\in \ran\left(B\right)$ is a right singular vector of $A$ with singular value $\InfNorm{A}$
	  
	  \item\label{item:restriction}
	  The restriction of $\frac{A }{\InfNorm{A}}$ to $\ran(B)$ is an isometry.
	\end{compactenum}
\end{proposition}

A proof of this proposition is provided in the Appendix. Applying it twice yields a similar result for the iterated version of H\"older's inequality
$\EinsNorm{ABC}\leq \InfNorm{A}\EinsNorm{B}\InfNorm{C}$ 
(for $A,B,C \in \L(\V)$).  

\begin{observation}[Saturation of iterated H\"older's inequality]\label{prop:Norm}
Let be $A,B,C \in \L\left(\V\right)$. \\
Then $ \EinsNorm{ABC} = \InfNorm{A}\EinsNorm{B}\InfNorm{C}$ if and only if
\begin{compactenum}[i)]
\item Every vector $x\in \ran\left(B\right)$ is a right singular vector of $A$ with singular value $\InfNorm{A}$ and 
\item Every vector $y \in \ran\left(B^\dagger\right)$ is a left singular vector of $C$ with singular value $\InfNorm{C}$.
\end{compactenum}
\end{observation}
Directly from the definition of the partial trace, it follows that 
\begin{equation}\label{eq:PartialTraceProperty}
\Tr[X(\1_{\W} \otimes B)] = \Tr\bigl[\Tr_\W(X)B\bigr]
\end{equation}
for all $X \in \L(\W\otimes\V)$ and $B \in \L(\V)$. 

\subsection{Proofs} \label{sec:proofs}
Now we use the previous observations to prove our theorems. 

\begin{proof}[Proof of Theorem \ref{thm:lowerBound}]
We recall the definition of the square norm, 
\begin{equation}
\DiaNorm{X} 
=  
\sup_{\substack{ A,B \in \L\left(\V \right)  \\ 
	             \ZweiNorm{A}=\ZweiNorm{B}=\sqrt{\dim(\V)}
	             }}
\EinsNorm{\left(\1_{\W }\otimes A \right) X \left( \1_{\W }\otimes B \right)} .
\end{equation}
For convenience we define $\tilde{X}_A\coloneqq \left( \1_{\W }\otimes A \right) X$ and $n\coloneqq \dim(\V)$. 
Using the unitary invariance of the nuclear norm, the fact that $\Tr[Y] \leq \EinsNorm{Y}$, and Eq.~\eqref{eq:PartialTraceProperty} we obtain
\begin{equation}	
\begin{aligned}
&\phantom{=.}
\EinsNorm{\left(\1_{\W }\otimes A \right) X \left( \1_{\W }\otimes B \right)}
\\
&\geq
\Tr\myleft[S_{\tilde{X}_A} \tilde{X}_A \left( \1_{\W }\otimes B \right)\myright]\\
&=
\Tr\myleft[ \sqrt{\tilde{X}_A^\dagger \tilde{X}_A} \left( \1_{\W } \otimes B\right) \myright]
\\
&= 
\Tr\myleft[ \Tr_\W\Bigl[\sqrt{\tilde{X}_A^\dagger \tilde{X}_A}\, \Bigr] B \myright] \, .
\end{aligned}
\end{equation}
Choosing $B \propto \Tr_\W \sqrt{\tilde{X}_A^\dagger\tilde{X}_A}$ (with proportionality constant chosen such that $\ZweiNorm{B} = \sqrt{n}$), which will turn out to yield the optimal value, we obtain
\begin{equation}\label{eq:intermediate_sqnorm_bound}
\DiaNorm{X}\geq \sup_{A\in \L\left(\V \right)}\left\lbrace \sqrt{n} \ZweiNorm{\Tr_\W \sqrt{\tilde{X}_A^\dagger \tilde{X}_A}}  
  , \, \ZweiNorm{A}=\sqrt{n}\right\rbrace \, .
\end{equation}
Using that $\sqrt{n} \ZweiNorm{Y} \geq \EinsNorm{Y}$ for all $Y \in \L(\V)$ yields
\begin{equation}\label{eq:2NormSaturationTilde}
\begin{aligned}
\sqrt{n}\ZweiNorm{\Tr_\W \sqrt{\tilde{X}_A^\dagger \tilde{X}_A}} 
&\geq 
\EinsNorm{\Tr_\W \sqrt{\tilde{X}_A^\dagger \tilde{X}_A}}
\\
&= 
\EinsNorm{ \left(\1_{\W }\otimes A \right) X } \, .
\end{aligned}
\end{equation}
Analogously as before, we proceed by using the unitary invariance of the $1$-norm and the sign matrix $S_X$ from the right to obtain 
\begin{equation}
\sqrt{n}\ZweiNormb{\Tr_\W \sqrt{\tilde{X}_A^\dagger \tilde{X}_A}} 
\geq 
\EinsNormb{ \left(\1_{\W }\otimes A \right) \sqrt{XX^\dagger} } \, .	
\end{equation}
We choose $A \propto \Tr_\W\sqrt{XX^\dagger}$ and combine it with the bound \eqref{eq:intermediate_sqnorm_bound} to finally arrive at 
\begin{equation}\label{eq:2NormSaturationNormal}
\begin{aligned}
\DiaNorm{X} &\geq  \sqrt{n}\ZweiNorm{\Tr_\W \sqrt{X X^\dagger }} 
\geq  
\EinsNorm{\Tr_\W \sqrt{X X^\dagger }} 
\\
&=
\EinsNorm{X} \, .
\end{aligned}
\end{equation}
The equality of both norms holds if and only if all inequalities are saturated. 
In particular we have that 
\begin{equation}
\sqrt{n}\ZweiNormb{\Tr_{\W }\sqrt{X X^\dagger}} 
= 
\EinsNormb{\Tr_{\W }\sqrt{X X^\dagger}}
\end{equation}
in the bound \eqref{eq:2NormSaturationNormal} if and only if $\Tr_{\W }\sqrt{XX^\dagger}$ is unitary up to a normalization factor (Observation~\ref{lem:Norm}). 
Since $\Tr_{\W }\sqrt{XX^\dagger}$ is also positive semidefinite we obtain
$\Tr_{\W }\sqrt{XX^\dagger} = \frac{\EinsNorm{X}}{n} \1$ as the only possible operator. 
From the bound \eqref{eq:2NormSaturationTilde} follows, with the same argument, that $\Tr_{\W }\sqrt{\tilde{X}_A^\dagger\tilde{X}_A} \propto \1$. This relation implies the second desired identity $\Tr_{\W }\sqrt{X^\dagger X}=\frac{\EinsNorm{X}}{n} \1$.
\end{proof}

\begin{proof}[Proof of Theorem \ref{thm:upperBound}]
$"\Rightarrow"$: 
We use the hypothesis and split the $1$-norm of products with an iterated H\"older's inequality in order to maximize separately over $A$ and $B$ in the definition of the square norm \eqref{eq:DiaNorm}, 
\begin{equation}
\begin{aligned}
n\EinsNorm{X} 
&=
\DiaNorm{X} 
\\
&\leq
  \sup_{ \substack{A,B\in\L\left(\V \right) \\
  		 \ZweiNorm{A}=\ZweiNorm{B}=\sqrt{n}
  		}}
  	\InfNorm{\1_{\W }\otimes A} \EinsNorm{X} \InfNorm{\1_{\W } \otimes B} 
\\
&=
n \EinsNorm{X} \, .
\end{aligned}
\end{equation}
The largest possible value is attained with rank $1$ matrices $A= \sqrt{n} \nu \psi^\dagger$, $B=\sqrt {n}\phi \mu^\dagger$ for normalized $\nu,\mu,\phi,\psi \in \V $. 
 Since the inequality is saturated we conclude with Observation~\ref{prop:Norm} that $X = Y \otimes \psi \phi^\dagger$ for an arbitrary non-zero operator $Y\in \L\left(\W \right)$, and the result follows.

$"\Leftarrow"$: 
From the hypothesis we have that $X= Y \otimes \psi \phi^\dagger$. Instead of maximizing in the definition of the square norm \eqref{eq:DiaNorm} we just choose $A\propto \psi\psi^\dagger$, $B \propto \phi\phi^\dagger$. This choice yields
\begin{equation}
\begin{aligned}
n \EinsNorm{X} &\geq 
\DiaNorm{X} 
\\
&\geq
n \EinsNorm{ \left(\1_{\W }\otimes \psi\psi^\dagger\right) \left( Y\otimes \psi\phi^\dagger \right) \left(\1_{\W }\otimes \phi\phi^\dagger \right) } 
\\
&=
n \EinsNorm{X}
\end{aligned}
\end{equation}
and the result follows.
\end{proof}

\bibliographystyle{IEEEtran}
\bibliography{martin}

\newpage
\section*{Appendix}
In this appendix we provide the proof of the saturation condition for H\"older's inequality.
\begin{proof}[Proof of Proposition~\ref{Obs:HoelderSat1}]
\hfill\newline
``\ref{item:every}) $\Rightarrow$ \ref{item:restriction})'':
$\frac{A}{\InfNorm{A}}$ maps an arbitrary $\phi \in \ran(B)$ from one singular basis to another. Hence, restricted to $\ran(B)$ it is an isometry.\\
``\ref{item:restriction}) $\Rightarrow$ \ref{item:tightness})'': 
We represent the isometry $\frac{A}{\InfNorm{A}}$ with a unitary matrix $U$ that fulfills $\InfNorm{A}UB =AB$.
 Using the unitary invariance of the nuclear norm, we obtain
$\EinsNorm{AB}=\InfNorm{A}\EinsNorm{UB}=\InfNorm{A}\EinsNorm{B} \, . $\\
``\ref{item:tightness}) $\Rightarrow$ \ref{item:every})'':
We denote the $i$-th largest singular value of a matrix $X$ with $\sigma_i(X)$.
 We upper bound the nuclear norm by splitting the singular values of the matrix as \cite[Problem~3.3.4]{Horn1991} 
 \begin{equation}
 	\EinsNorm{AB}=\sum_{i=1}^n\sigma_i(AB) \leq \sum_{i=1}^n\sigma_i(A)\sigma_i(B) \, .
 \end{equation}
Upper bounding the singular values of $A$ by the largest singular value yields 
\begin{equation}
\begin{aligned}
	\sum_{i=1}^n\sigma_i(A)\sigma_i(B) 
	&\leq 
	\sum_{i=1}^{\rank(B)}\InfNorm{A}\sigma_i(B)
	\\
	&=\InfNorm{A}\EinsNorm{B}\, .
\end{aligned}
\end{equation}
 By assumption, every inequality is saturated. In particular, the largest $\rank(B)$ singular values of $A$ are all equal. 
 Furthermore, $\sigma_i(AB)=\InfNorm{A}\sigma_i(B)$ and, thus, all vectors $x \in \ran(B)$ are right singular vectors of $A$ with singular value $\InfNorm{A}$.
\end{proof}

\end{document}

%% file: mymath.tex
\newcommand{\myleft}{\mathopen{}\mathclose\bgroup\left}
\newcommand{\myright}{\aftergroup\egroup\right}
\DeclareMathOperator{\Tr}{Tr}
\DeclareMathOperator{\ran}{ran}
\DeclareMathOperator{\rank}{rank}

\newcommand{\1}{\mathds{1}}

\newcommand{\mc}[1]{\mathcal{#1}}

\newcommand{\V}{\mc{V}}
\renewcommand{\W}{\mc{W}}

\newcommand{\norm}[1]{\left\Vert #1 \right\Vert}
\newcommand{\normb}[1]{\bigl\Vert #1 \bigr\Vert}
\newcommand{\dnorm}[1]{\norm{#1}_\diamond} 

\newcommand{\EinsNorm}[1]{\norm{#1}_1}

\newcommand{\EinsNormb}[1]{\normb{#1}_1}

\newcommand{\ZweiNorm}[1]{\norm{#1}_2}

\newcommand{\ZweiNormb}[1]{\normb{#1}_2}

\newcommand{\DiaNorm}[1]{\norm{#1}_\square}
\newcommand{\InfNorm}[1]{\norm{#1}_\infty}

\newcommand{\argdot}{{\,\cdot\,}}

\renewcommand{\L}{\operatorname{\mathrm{L}}}
